\documentclass[10pt]{bmc_article}
\usepackage{amsfonts, amssymb, amsmath, amsthm, bm}
\usepackage{url, ifthen, multicol, color, upgreek, cite, rotating, endnotes}
\usepackage{hyperref}
\newcommand{\Rbb}{\mathbb{R}} \newcommand{\Cbb}{\mathbb{C}}
 
 \newcommand{\Prob}{{\mathbb{P}}}
\newcommand{\scp}[2]{\langle #1, #2 \rangle}
\newtheorem{theorem}{Theorem} \newtheorem{definition}{Definition}
 \newtheorem{lemma}{Lemma}
 \newtheorem{corollary}{Corollary}

\renewcommand{\leq}{\leqslant} \renewcommand{\geq}{\geqslant}
\newcommand{\norm}[1]{\|#1\|} \newcommand{\abs}[1]{\left| #1 \right|}
\newcommand{\ma}[1]{\mathsf{#1}}

\newboolean{publ}
\newenvironment{bmcformat}{\baselineskip14pt\setboolean{publ}{false}}

\begin{document}
\begin{bmcformat}

%
%
%
%
\title{Universal and efficient compressed sensing by spread spectrum and application to realistic Fourier imaging techniques}

\author
{%
Gilles Puy$^{1, 2}$\email{Gilles Puy\correspondingauthor - gilles.puy@epfl.ch}\and
Pierre Vandergheynst$^{1}$\email{Pierre Vandergheynst - pierre.vandergheynst@epfl.ch}\and
R\'{e}mi Gribonval$^{3}$\email{R\'{e}mi Gribonval - remi.gribonval@inria.fr}\and
Yves Wiaux$^{1, 3, 4}$\email{Yves Wiaux - yves.wiaux@epfl.ch}
}

\address
{%
\iid(1)Institute of Electrical Engineering,  Ecole Polytechnique F{\'e}d{\'e}rale de Lausanne (EPFL), CH-1015 Lausanne, Switzerland.\\
\iid(2)Institute of the Physics of Biological Systems,  Ecole Polytechnique F{\'e}d{\'e}rale de Lausanne (EPFL), CH-1015 Lausanne, Switzerland.\\
\iid(3)Centre de Recherche INRIA Rennes-Bretagne Atlantique, F-35042 Rennes cedex, France.\\
\iid(4)Institute of Bioengineering,  Ecole Polytechnique F{\'e}d{\'e}rale de Lausanne (EPFL), CH-1015 Lausanne, Switzerland.\\
\iid(5)Department of Radiology and Medical Informatics, University of Geneva (UniGE), CH-1211 Geneva, Switzerland.
}

\maketitle

\begin{abstract}
We advocate a compressed sensing strategy that consists of multiplying the signal of interest by a wide bandwidth modulation before projection onto randomly selected vectors of an orthonormal basis. Firstly, in a digital setting with random modulation, considering a whole class of sensing bases including the Fourier basis, we prove that the technique is \emph{universal} in the sense that the required number of measurements for accurate recovery is optimal and independent of the sparsity basis. This universality stems from a drastic decrease of coherence between the sparsity and the sensing bases, which for a Fourier sensing basis relates to a spread of the original signal spectrum by the modulation (hence the name ``spread spectrum''). The approach is also \emph{efficient} as sensing matrices with fast matrix multiplication algorithms can be used, in particular in the case of Fourier measurements. Secondly, these results are confirmed by a numerical analysis of the phase transition of the $\ell_1$-minimization problem. Finally, we show that the spread spectrum technique remains effective in an analog setting with chirp modulation for application to realistic Fourier imaging. We illustrate these findings in the context of radio interferometry and magnetic resonance imaging.
\end{abstract}

\ifthenelse{\boolean{publ}}{\begin{multicols}{2}}{}
%

%
%
%
%
\section{Introduction}
\label{sec:introduction}

In this section we concisely recall some basics of compressed sensing, emphasizing on the role of mutual coherence between the sparsity and sensing bases. We discuss the interest of improving the standard acquisition strategy in the context of Fourier imaging techniques such as radio interferometry and magnetic resonance imaging (MRI). Finally, we highlight the main contributions of our work, advocating a universal and efficient compressed sensing strategy coined spread spectrum, and describe the organization of this article.

\subsection{Compressed sensing basics}

Compressed sensing is a recent theory aiming at merging data acquisition and compression \cite{candes06a, donoho06, baraniuk07, candes06b, donoho09, candes07, rauhut10}. It predicts that sparse or compressible signals can be recovered from a small number of linear and non-adaptative measurements. In this context, Gaussian and Bernouilli random matrices, respectively with independent standard normal and $\pm 1$ entries, have encountered a particular interest as they provide optimal conditions in terms of the number of measurements needed to recover sparse signals \cite{baraniuk07, candes06b, donoho09}. However, the use of these matrices for real-world applications is limited for several reasons: no fast matrix multiplication algorithm is available, huge memory requirements for large scale problems, difficult implementation on hardware, etc.

Let us consider $s$-sparse digital signals $\bm x \in \Cbb^N$ in an orthonormal basis $\ma \Psi = (\bm \psi_1, ..., \bm \psi_N) \in \Cbb^{N \times N}$. The decomposition of $\bm x$ in this basis is denoted $\bm \alpha= \left(\alpha_i\right)_{1\leq i \leq N} \in \Cbb^N$, $\bm{\alpha}= \ma{\Psi}^* \bm{x}$ ($\cdot^*$ denotes the conjugate transpose), and contains $s$ non-zero entries. The original signal $\bm{x}$ is then probed by projection onto $m$ randomly selected vectors of another orthonormal basis $\ma{\Phi} = (\bm \phi_1, ..., \bm \phi_N) \in \Cbb^{N \times N}$. The indices $\Omega = \left\{ l_1, \ldots, l_m\right\}$ of the selected vectors are chosen independently and uniformly at random from $\left\{1, \ldots, N\right\}$. We denote $\ma \Phi^*_\Omega$ the $m \times N$ matrix made of the selected rows of $\ma \Phi^*$. The measurement vector $\bm y \in \Cbb^m$ thus reads as
\begin{eqnarray}
\label{eq:measurement model}
\bm y = \ma{A}_\Omega \, \bm \alpha \text{ with } \ma{A}_\Omega = \ma{\Phi}^*_\Omega \ma{\Psi} \in \Cbb^{m \times N}.
\end{eqnarray}
We also denote $\ma{A} = \ma{\Phi}^*\ma{\Psi} \in \Cbb^{N \times N}$. Finally, we aim at recovering $\bm{\alpha}$ by solving the $\ell_1$-minimization problem
\begin{eqnarray}
\label{eq:BP}
\bm{\alpha}^\star = \arg\min_{\bm{\alpha} \in \Cbb^N} \norm{\bm{\alpha}}_1 \text{  subject to  }  \bm y=\ma{A}_\Omega \bm{\alpha},
\end{eqnarray}
where $\norm{\bm{\alpha}}_1 = \sum_{i=1}^N \abs{\alpha_i}$ ($\abs{\cdot}$ denotes the complex magnitude). The reconstructed signal $\bm{x}^\star$ satisfies $\bm{x}^\star = \ma{\Psi}\bm{\alpha}^\star$.

The theory of compressed sensing already demonstrates that a small number $m\ll N$ of random measurements are sufficient for an accurate and stable reconstruction of $\bm x$ \cite{candes07, rauhut10}. However, the recovery conditions depend on the mutual coherence $\mu$ between $\ma \Phi$ and $\ma \Psi$. This value is a similarity measure between the sensing and sparsity bases. It is defined as $\mu = \max_{1\leq i,j \leq N} \left| \scp{\bm \phi_i}{\bm \psi_j}\right|$ and satisfies $N^{-1/2}\leq\mu\leq1$. The performance is optimal when the bases are perfectly incoherent, i.e. $\mu = N^{-1/2}$, and unavoidably decreases when $\mu$ increases.

%
%
%
%
\subsection{Fourier imaging applications and mutual coherence}
The dependence of performance on the mutual coherence $\mu$ is a key concept in compressed sensing. It has significant implications for Fourier imaging applications, in particular radio interferometry or magnetic resonance imaging (MRI), where signals are probed in the orthonormal Fourier basis. In radio interferometry, one of the main challenges is to reconstruct accurately the original signal from a limited number of accessible measurements \cite{hogbom74, cornwell85, wiaux09, wenger10, li11}. In MRI, accelerating the acquisition process by reducing the number of measurements is of huge interest { in, for example, static and dynamic imaging \cite{lustig07, jung07, gamper08, jung09, usman11}, parallel MRI \cite{liang09a, lustig10, otazo10}, or MR spectroscopic imaging \cite{hu08, hu10, larson10}}. The theory of compressed sensing shows that Fourier acquisition is the best sampling strategy when signals are sparse in the Dirac basis. The sensing system is indeed optimally incoherent. Unfortunately, natural signals are usually rather sparse in multi-scale bases, e.g., wavelet bases, which are coherent with the Fourier basis. Many measurements are thus needed to reconstruct accurately the original signal. In the perspective of accessing better performance, sampling strategies that improve the incoherence of the sensing scheme should be considered.

\subsection{Main contributions and organization}
In the present work, we advocate a compressed sensing strategy coined spread spectrum that consists of a wide bandwidth pre-modulation of the signal $\bm x$ before projection onto randomly selected vectors of an orthonormal basis. In the particular case of Fourier measurements, the pre-modulation amounts to a convolution in the Fourier domain which spreads the power spectrum of the original signal $\bm x$ (hence the name ``spread spectrum''), while preserving its norm. Equivalently, this spread spectrum phenomenon acts on each sparsity basis vector describing $\bm x$ so that information of each of them is accessible whatever the Fourier coefficient selected. This effect implies a decrease of coherence between the sparsity and sensing bases and enables an enhancement of the reconstruction quality.

In Section \ref{sec:discrete spread spectrum}, we study the spread spectrum technique in a digital setting for arbitrary pairs of sensing and sparsity bases $\left(\ma \Phi, \ma \Psi \right)$. We consider a digital pre-modulation  $\bm c = (c_l)_{1\leq l \leq N} \in \Cbb^N$ with $\abs{c_l} = 1$ and random phases identifying a random Rademacher or Steinhaus sequence. We show that the recovery conditions do not depend anymore on the coherence of the system but on a new parameter $\beta\left(\ma{\Phi}, \ma{\Psi}\right)$ called \emph{modulus-coherence} and defined as
\begin{eqnarray}
\label{eq:modulated coherence}
\beta\left(\ma{\Phi}, \ma{\Psi}\right) = \max_{1\leq i, j \leq N} \sqrt{\sum_{k=1}^N \left| \phi_{ki}^* \psi_{kj} \right|^2},
\end{eqnarray}
where $\phi_{ki}$ and $\psi_{kj}$ are respectively the $k^{\text{th}}$ entries of the vectors $\bm{\phi}_{i}$ and $\bm{\psi}_{j}$. We then show that this parameter reaches its optimal value $\beta\left(\ma{\Phi}, \ma{\Psi}\right) = N^{-1/2}$ whatever the sparsity basis $\Psi$, for particular sensing matrices $\Phi$ including the Fourier matrix, thus providing \emph{universal} recovery performances.  It is also \emph{efficient} as sensing matrices with fast matrix multiplication algorithms can be used, thus reducing the need in memory requirement and computational power. In Section \ref{sec:Numerical simulations}, these theoretical results are confirmed numerically through an analysis of the { empirical} phase transition of the $\ell_1$-minimization problem for different pairs of sensing and sparsity bases. In Section \ref{sec:analog spread spectrum}, we show that the spread spectrum technique remains effective in an analog setting with chirp modulation for application to realistic Fourier imaging, and illustrate these findings in the context of radio interferometry { and MRI}. Finally, we conclude in Section \ref{sec:conclusion}.

{ In the context of compressed sensing, the spread spectrum technique was already briefly introduced by the authors for compressive sampling of pulse trains in \cite{naini09}, applied to radio interferometry in \cite{wiaux09b, wiaux09a} and to MRI in \cite{puy10a, wiaux10, puy11a, puy11b}. This paper provides theoretical foundations for this technique, both in the digital and analog settings. Note that other acquisition strategies can be related to the spread spectrum technique as discussed in Section \ref{sub:related}.

Let us also acknowledge that spread spectrum techniques are very popular in telecommunications. For example, one can cite the direct sequence spread spectrum (DSSS) and the frequency hopping spread spectrum (FHSS) techniques. The former is sometimes used over wireless local area networks, the latter is used in Bluetooth systems \cite{schiller03}. In general, spread spectrum techniques are used for their robustness to narrowband interference and also to establish secure communications.}

%
%
%
%
\section{Compressed sensing by spread spectrum}
\label{sec:discrete spread spectrum}

In this section, we first recall the standard recovery conditions of sparse signals randomly sampled in a bounded orthonormal system. These recovery results depend on the mutual coherence $\mu$ of the system. Hence, we study the effect of a random pre-modulation on this value and deduce recovery conditions for the spread spectrum technique. We finally show that the number of measurements needed to recover sparse signals becomes universal for a family of sensing matrices $\ma \Phi$ which includes the Fourier basis.

%
%
%
%
\subsection{Recovery results in a bounded orthonormal system}

For the setting presented in Section \ref{sec:introduction}, the theory of compressed sensing already provides sufficient conditions on the number of measurements needed to recover the vector $\bm{\alpha}$ from the measurements $\bm{y}$ by solving the $\ell_1$-minimization problem (\ref{eq:BP}) \cite{candes07, rauhut10}.
\begin{theorem}[\cite{rauhut10}, Theorem 4.4]
\label{th:standard uniform recovery}
Let $\ma{A} = \ma{\Phi}^*\ma{\Psi} \in \Cbb^{N \times N}$, $\mu = \max_{1\leq i,j \leq N} \left| \scp{\bm \phi_i}{\bm \psi_j} \right|$, $\bm{\alpha} \in \Cbb^N$ be an $s$-sparse vector, { $\Omega = \left\{ l_1, \ldots, l_m\right\}$ be a set of $m$ indices chosen independently and uniformly at random from $\left\{1, \ldots, N\right\}$}, and $\bm{y} = \ma{A}_\Omega \bm{\alpha} \in \Cbb^{m}$. For some universal constants $C>0$ and $\gamma>1$, if
\begin{eqnarray}
\label{eq:standard uniform recovery}
m \geq C N \mu^2 s \log^4(N),
\end{eqnarray}
then $\bm \alpha$ is the unique minimizer of the $\ell_1$-minimization problem (\ref{eq:BP}) with probability at least $1-N^{-\gamma \log^3(N)}$.
\end{theorem}

Let us acknowledge that even if the measurements are corrupted by noise or if $\bm{\alpha}$ is non-exactly sparse, the theory of compressed sensing also shows that the reconstruction obtained by solving the $\ell_1$-minimization problem remains accurate and stable:
\begin{theorem}[\cite{rauhut10}, Theorem 4.4]
\label{pr:stability noise}
Let $\ma{A} = \ma{\Phi}^*\ma{\Psi}$, {$ \Omega = \left\{ l_1, \ldots, l_m\right\}$ be a set of $m$ indices chosen independently and uniformly at random from $\left\{1, \ldots, N\right\}$}, and ${\cal T}_s\left(\bm{\alpha}\right)$ be the best $s$-sparse approximation of the (possibly non-sparse) vector $\bm{\alpha} \in \Cbb^N$. Let the noisy measurements $\bm{y} = \ma{A}_\Omega \, \bm{\alpha} + \bm{n} \in \Cbb^m$ be given with $\norm{\bm{n}}_2^2 =  \sum_{i=1}^m \abs{n_i}^2 \leq \eta^2$, $\eta \geq 0$. For some universal constants $D, E >0$ and $\gamma>1$, if relation (\ref{eq:standard uniform recovery}) holds, then the solution $\bm{\alpha}^\star$ of the $\ell_1$-minimization problem
\begin{eqnarray}
\label{eq:BPDN}
\bm{\alpha}^\star = \arg\min_{\bm{\alpha} \in \Cbb^N} \norm{\bm{\alpha}}_1 \textnormal{ subject to }  \norm{\bm{y}-\ma{A}_\Omega \bm{\alpha}}_2 \leq \eta,
\end{eqnarray}
satisfies
\begin{eqnarray}
\norm{\bm{\alpha} - \bm{\alpha}^\star}_2 \leq D \, \frac{\norm{\bm{\alpha}-{\cal T}_s\left(\bm{\alpha}\right)}_1}{s^{1/2}} + E \, \eta,
\end{eqnarray}
with probability at least $1-N^{-\gamma \log^3(N)}$.
\end{theorem}

In the above theorems, the role of the mutual coherence $\mu$ is crucial as the number of measurements needed to reconstruct $\bm x$ scales quadratically with its value. In the worst case where $\ma{\Phi}$ and $\ma{\Psi}$ are identical, $\mu = 1$ and the signal $\bm x$ is probed in a domain where it is also sparse. According to relation (\ref{eq:standard uniform recovery}), the number of measurements necessary to recover $\bm x$ is of order $N$. This result is actually very intuitive. For an accurate reconstruction of signals sampled in their sparsity domain, all the non-zero entries need to be probed. It becomes highly probable when $m \simeq N$. On the contrary, when $\ma{\Phi}$ and $\ma{\Psi}$ are as incoherent as possible, i.e., $\mu = N^{-1/2}$, the energy of the sparsity basis vectors spreads equally over the sensing basis vectors. Consequently, whatever the sensing basis vector selected, one always gets information of all the sparsity basis vectors describing the signal $\bm x$, therefore reducing the need in the number of measurements. This is confirmed by relation (\ref{eq:standard uniform recovery}) which shows that the number of measurements is of the order of $s$ when $\mu = N^{-1/2}$. To achieve much better performance when the mutual coherence is not optimal, one would naturally try to modify the measurement process to achieve a better global incoherence. We will see in the next section that a simple random pre-modulation is an efficient way to achieve this goal whatever the sparsity matrix $\ma{\Psi}$.

%
%
%
%
\subsection{Pre-modulation effect on the mutual coherence}
\label{sub:modulation effect on coherence}

The spread spectrum technique consists of pre-modulating the signal $\bm x$ by a wide-band signal $\bm c = (c_l)_{1\leq l \leq N} \in \Cbb^N$, with $\abs{c_l} = 1$ and random phases, before projecting the resulting signal onto $m$ vectors of the basis $\Phi$. The measurement vector $\bm y$ thus satisfies
\begin{eqnarray}
\label{eq:spread measurement matrix}
\bm y = \ma{A}^{\rm c}_\Omega \, \bm \alpha \text{ with } \ma{A}^{\rm c}_\Omega = \ma{\Phi^*_\Omega C\Psi} \in \Cbb^{m \times N},
\end{eqnarray}
where the additional matrix $\ma{C} \in \Rbb^{N \times N}$ stands for the diagonal matrix associated to the sequence $\bm c$.

In this setting, the matrix $\ma{A}^{\rm c}$ is orthonormal. Therefore, the recovery condition of sparse signals sampled with this matrix depends on the mutual coherence $\mu = \max_{1\leq i,j \leq N} \left| \scp{\bm \phi_i}{\ma{C}\,\bm \psi_j}\right|$. With a pre-modulation by a random Rademacher or Steinhaus sequence, Lemma \ref{lm:coherence bound} shows that the mutual coherence $\mu$ is essentially bounded by the modulus-coherence $\beta\left(\ma{\Phi}, \ma{\Psi}\right)$ defined in equation (\ref{eq:modulated coherence}).

\begin{lemma}
\label{lm:coherence bound}
Let { $\epsilon \in (0, 1)$}, $\bm{c} \in \Cbb^{N}$ be a random Rademacher or Steinhaus sequence and $\ma{C} \in \Cbb^{N \times N}$ be the associated diagonal matrix. Then, the mutual coherence $\mu= \max_{1\leq i,j \leq N} \left| \scp{\bm \phi_i}{\ma{C}\,\bm \psi_j}\right|$ satisfies
\begin{eqnarray}
\label{eq:coherence bound}
\mu \leq \beta\left(\ma{\Phi}, \ma{\Psi}\right) \sqrt{2\log\left(2N^2/\epsilon\right)},
\end{eqnarray}
with probabilty at least $1 - \epsilon$.
\end{lemma}

The proof of Lemma \ref{lm:coherence bound} relies on a simple application of the Hoeffding's inequality and the union bound.
\begin{proof}
We have $\scp{\bm \phi_i}{\ma{C} \, \bm \psi_j} = \sum_{k=1}^N c_k \phi_{ki}^* \psi_{kj} = \sum_{k=1}^N c_k a_k^{ij}$, where $a_k^{ij} = \phi_{ki}^* \psi_{kj}$. An application of the Hoeffding's inequality shows that
\begin{eqnarray*}
\Prob \left( \left| \scp{\bm{\phi}_i}{\ma{C}\bm{\psi}_j} \right| > u \right) \leq 2 \exp\left(-\frac{u^2}{2 \norm{\bm{a}^{ij}}_2^2}\right),
\end{eqnarray*}
for all $u>0$ and $1 \leq i, j \leq N$, with $\norm{\bm{a}^{ij}}_2^2 = \sum_{k=1}^N \left| a_k^{ij} \right|^2$.
The union bound then yields
\begin{eqnarray*}
\Prob \left( \mu > u \right) 
& \leq & \sum_{1 \leq i, j \leq N} \Prob \left( \left| \scp{\bm \phi_i}{\bm c \cdot \bm \psi_j} \right| > u \right) \\
& \leq &  2 \sum_{1 \leq i, j \leq N} \exp\left(-\frac{u^2}{2 \norm{\bm{a}^{ij}}_2^2}\right),
\end{eqnarray*}
for all $u>0$. As $\beta^2\left(\ma{\Phi}, \ma{\Psi}\right) = \max_{1\leq i, j \leq N} \sum_{k=1}^N  \left| a_k^{ij} \right|^2$ then $\norm{\bm{a}^{ij}}_2^2 \leq \beta^2\left(\ma{\Phi}, \ma{\Psi}\right)$ for all $1 \leq i, j \leq N$, and the previous relation becomes
\begin{eqnarray*}
\Prob \left( \mu > u \right) & \leq & 2 N^2 \exp\left(-\frac{u^2}{2 \beta^2\left(\ma{\Phi}, \ma{\Psi}\right)}\right),
\end{eqnarray*}
for all $u>0$. Taking $u=\sqrt{2\beta^2\left(\ma{\Phi}, \ma{\Psi}\right)\log\left(2N^2/\epsilon\right)}$ terminates the proof.
\end{proof}
%

%
%
%
%
\subsection{Sparse recovery with the spread spectrum technique}

Combining Theorem \ref{th:standard uniform recovery} with the previous estimate on the mutual coherence, we can state the following theorem:

\begin{theorem}
\label{th:spread spectrum uniform recovery}
Let $\bm{c} \in \Cbb^{N}$, with $N\geq2$, be a random Rademacher or Steinhaus sequence, $\ma{C} \in \Cbb^{N \times N}$ be the diagonal matrix associated to $\bm{c}$, $\bm{\alpha} \in \Cbb^N$ be an $s$-sparse vector, { $\Omega = \left\{ l_1, \ldots, l_m\right\}$ be a set of $m$ indices chosen independently and uniformly at random from $\left\{1, \ldots, N\right\}$}, and $\bm{y} = \ma{A}^{\rm c}_\Omega \bm{\alpha} \in \Cbb^{m}$, with $\ma A^{\rm c} = \ma{\Phi^* C\Psi}$. For some constants $0 < \rho < \log^3(N)$ and $C_{\rho} > 0$, if
\begin{eqnarray}
\label{eq:spread spectrum uniform recovery}
m \geq C_{\rho} \, N \beta^2\left(\ma{\Phi}, \ma{\Psi}\right) s \log^5(N),
\end{eqnarray}
then $\bm \alpha$ is the unique minimizer of the $\ell_1$-minimization problem (\ref{eq:BP}) with probability at least $1-{\cal O}\left(N^{-\rho}\right)$.
\end{theorem}

\begin{proof}
{ It is straightforward to check that $\ma{C}^*\ma{C} = \ma{C}\ma{C}^* = \ma{I}$ where $\ma{I}$ is the identity matrix. The matrix $\ma{A}^{\rm c} = \ma{\Phi^* C\Psi}$ is thus orthonormal and Theorem \ref{th:standard uniform recovery} applies. To keep the notations simple, let us denote $F$ the event of failure of the $\ell_1$-minimization problem (\ref{eq:BP}), $X$ the event $m \geq C N \mu^2 s \log^4(N)$, and $Y$ the event $\beta\left(\ma{\Phi}, \ma{\Psi}\right) \sqrt{2\log\left(2N^2/\epsilon\right)} \geq \mu$. According to Theorem \ref{th:standard uniform recovery} and Lemma \ref{lm:coherence bound}, the probability of $F$ given $X$ satisfies $\Prob(F \vert X) \leq N^{-\gamma \log^3(N)}$ and the probability of $Y$ satisfies $\Prob(Y) \geq 1 - \epsilon$.

We will shortly see that for a proper choice of $\epsilon$, when condition (\ref{eq:spread spectrum uniform recovery}) holds, we have 
\begin{eqnarray}
\label{eq:cond1m}
m \geq 2C \, N \beta^2\left(\ma{\Phi}, \ma{\Psi}\right) s \log\left(2N^2/\epsilon\right) \log^4(N).
\end{eqnarray}
Using this fact, we compute the probability of failure $\Prob(F)$ of the $\ell_1$ minimization problem. We start by noticing that
\begin{eqnarray*}
\Prob(F) = \Prob(F \vert X) \Prob(X) + \Prob(F \vert X^c) \Prob(X^c) \leq \Prob(F \vert X) + \Prob(X^c) \leq N^{-\gamma \log^3(N)} + \Prob(X^c),
\end{eqnarray*}
where $X^c$ denotes the complement of event $X$. In the first inequality, the probability $\Prob(X)$ and $\Prob(F \vert X^c)$ are saturated to $1$. One can also note that if $\beta\left(\ma{\Phi}, \ma{\Psi}\right) \sqrt{2\log\left(2N^2/\epsilon\right)} \geq \mu$, i.e., $Y$ occurs, condition (\ref{eq:cond1m}) implies that $m \geq C N \mu^2 s \log^4(N)$, i.e., $X$ occurs. Therefore $\Prob(X \vert Y) = 1$, $\Prob(X^c \vert Y) = 0$ and 
\begin{eqnarray*}
\Prob(X^c) = \Prob(X^c \vert Y) \Prob(Y) + \Prob(X^c \vert Y^c) \Prob(Y^c) = \Prob(X^c \vert Y^c) \Prob(Y^c) \leq \Prob \left(Y^c\right) \leq \epsilon.
\end{eqnarray*}

The probability of failure is thus bounded above by $N^{-\gamma \log^3(N)} + \epsilon$. Consequently, if condition (\ref{eq:cond1m}) holds with $\epsilon = N^{-\rho}$ and $0 < \rho < \log^3(N)$, $\bm \alpha$ is the unique minimizer of the $\ell_1$-minimization problem (\ref{eq:BP}) with probability at least $1-{\cal O}(N^{-\rho})$. 

Finally, noticing that for $\epsilon = N^{-\rho}$ with $N \geq 2$, condition (\ref{eq:cond1m}) always holds when condition (\ref{eq:spread spectrum uniform recovery}), with $C_\rho= 2(3+\rho)C$, is satisfied, terminates the proof.}
\end{proof}

Note that relation (\ref{eq:spread spectrum uniform recovery}) also ensures the stability of the spread spectrum technique relative to noise and compressibility by combination of Theorem \ref{pr:stability noise} and Lemma \ref{lm:coherence bound}.

%
%
%
%
\subsection{Universal sensing bases with ideal modulus-coherence}

Theorem \ref{th:spread spectrum uniform recovery} shows that the performance of the spread spectrum technique is driven by the modulus-coherence $\beta\left(\ma{\Phi}, \ma{\Psi}\right)$. In general the spread spectrum technique is not universal and the number of measurements required for accurate reconstructions depends on the value of this parameter. 

\begin{definition} (Universal sensing basis)
An orthonormal basis $\ma{\Phi} \in \Cbb^{N \times N}$ is called a universal sensing basis if all its entries $\phi_{ki}$, $1\leq k, i \leq N$, are of equal complex magnitude.
\end{definition}

For universal sensing bases, e.g. the Fourier transform or the Hadamard transform, we have $\left| \phi_{ki} \right| = N^{-1/2}$ for all $1 \leq k, i \leq N$. It follows that $\beta\left(\ma{\Phi}, \ma{\Psi}\right) = N^{-1/2}$ and $\mu \simeq N^{-1/2}$, i.e. its optimal value up to a logarithmic factor, whatever the sparsity matrix considered! For such sensing matrices, the spread spectrum technique is thus a simple and efficient way to render a system incoherent independently of the sparsity matrix.

\begin{corollary} (Spread spectrum universality)
\label{co:spread spectrum universality}
Let $\bm{c} \in \Cbb^{N}$, with $N \geq 2$, be a random Rademacher or Steinhaus sequence, $\ma{C} \in \Cbb^{N \times N}$ be the diagonal matrix associated to $\bm{c}$, $\bm{\alpha} \in \Cbb^N$ be an $s$-sparse vector, { $\Omega = \left\{ l_1, \ldots, l_m\right\}$ be a set of $m$ indices chosen independently and uniformly at random from $\left\{1, \ldots, N\right\}$}, and $\bm{y} = \ma{A}^{\rm c}_\Omega \bm{\alpha} \in \Cbb^{m}$, with $\ma{A}^{\rm c} = \ma{\Phi^* C\Psi}$. For some constants $0 < \rho < \log^3(N)$, $C_{\rho} > 0$, and universal sensing bases $\ma{\Phi} \in \Cbb^{N \times N}$, if
\begin{eqnarray}
\label{eq:spread spectrum universality}
m \geq C_{\rho} \, s \log^5(N),
\end{eqnarray}
then $\bm \alpha$ is the unique minimizer of the $\ell_1$-minimization problem (\ref{eq:BP}) with probability at least $1-{\cal O}\left(N^{-\rho}\right)$.
\end{corollary}

For universal sensing bases, the spread spectrum technique is thus \emph{universal}: the recovery condition does not depend on the sparsity basis and the number of measurements needed to reconstruct sparse signals is optimal in the sense that it is reduced to the sparsity level $s$. The technique is also \emph{efficient} as the pre-modulation only requires a sample-by-sample multiplication between $\bm{x}$ and $\bm{c}$. Furthermore, fast multiplication matrix algorithms are available for several universal sensing bases such as the Fourier or Hadamard bases.

In light of Corollary \ref{co:spread spectrum universality}, one can notice that sampling sparse signals in the Fourier basis is a universal encoding strategy whatever the sparsity basis $\ma{\Psi}$ - even if the original signal is itself sparse in the Fourier basis! We will confirm these results experimentally in Section \ref{sec:Numerical simulations}.

%
%
%
%
\subsection{Related work}
\label{sub:related}

Let us acknowledge that the techniques proposed in \cite{do08, do11, romberg09, tropp10, krahmer11, tropp11} can be related to the spread spectrum technique. The benefit of a random pre-modulation in the measurement system is already briefly suggested in \cite{do08}. { The proofs of the claims presented in that conference paper have very recently been accepted for publication in \cite{do11} during the review process of this paper. The authors obtain similar recovery results as those presented here.} In \cite{romberg09}, the author proposes to convolve the signal $\bm x$ with a random waveform and randomly under-sample the result in time-domain. The random convolution is performed through a random pre-modulation in the Fourier domain and the signal thus spreads in time-domain. In our setting, this method actually corresponds to taking $\ma \Phi$ as the Fourier matrix and $\ma \Psi$ as the composition of the Fourier matrix and the initial sparsity matrix. In \cite{tropp10}, the authors propose a technique to sample signals sparse in the Fourier domain. They first pre-modulate the signal by a random sequence, then apply a low-pass antialiasing filter, and finally sample it at low rate. Finally, random pre-modulation is also used in \cite{krahmer11} and \cite{tropp11} but for dimension reduction and low dimensional embedding.

We recover similar results, albeit in a different way. We also have a more general interpretation. In particular, we proved that changing the sensing matrix from the Fourier basis to the Hadamard does not change the recovery condition (\ref{eq:spread spectrum universality}). 

%
%
%
%
\section{Numerical simulations}
\label{sec:Numerical simulations}

In this section, we confirm our theoretical predictions by showing, through a numerical analysis of the phase transition of the $\ell_1$-minimization problem, that the spread spectrum technique is universal for the Fourier and Hadamard sensing bases.

\begin{figure}
\centering
\includegraphics[width=14cm,keepaspectratio]{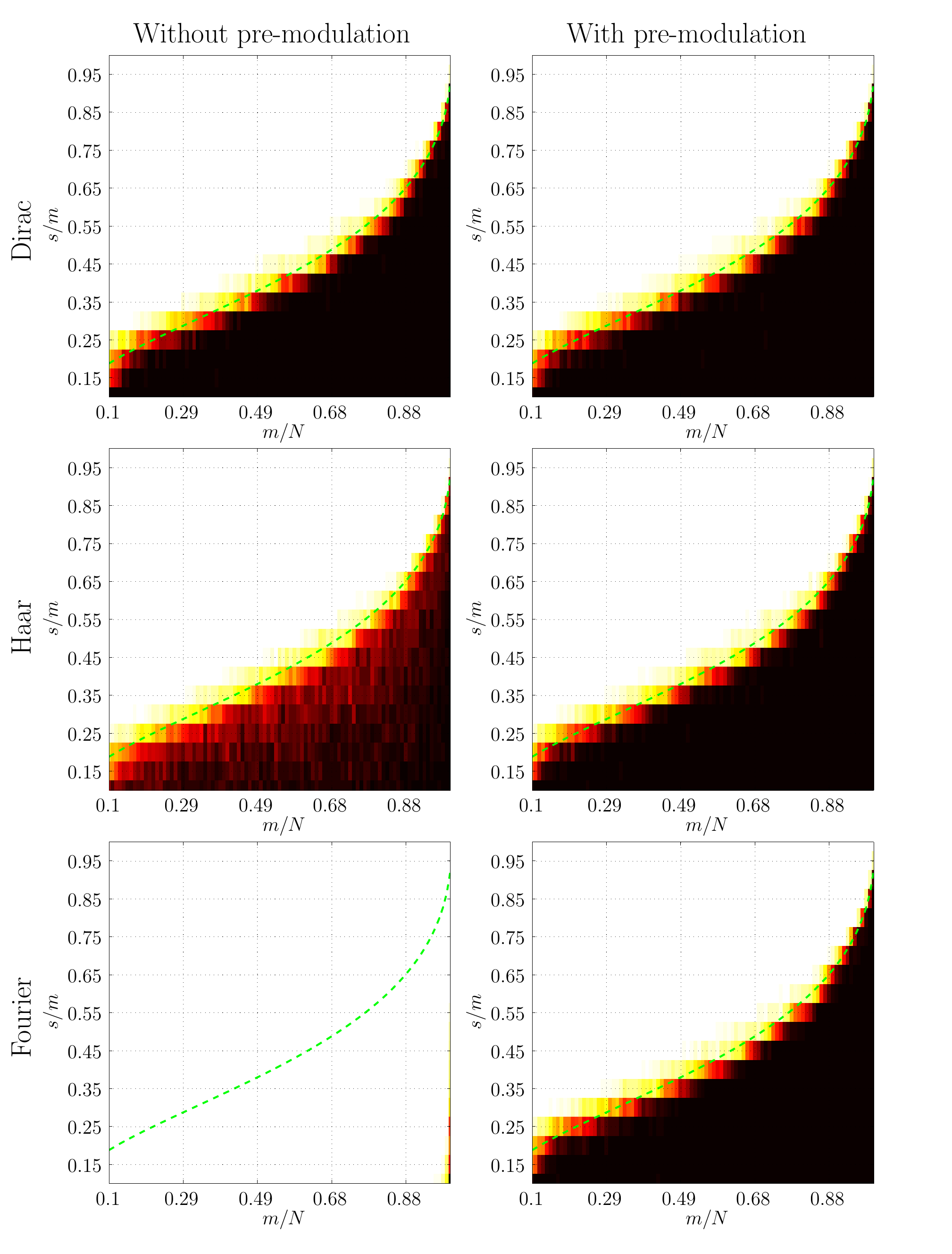}
\caption{\label{fig:phase transition fourier}Phase transition of the $\ell_1$-minimization problem for different sparsity bases and random selection of \textbf{Fourier} measurements without (left panels) and with (right panels) random modulation. The sparsity bases considered are the Dirac basis (top), the Haar wavelet basis (center), and the Fourier basis (bottom). The dashed green line indicates the phase transition of Donoho-Tanner \cite{donoho09}. The color bar goes from white to black indicating a probability of recovery from $0$ to $1$.}
\end{figure} 
\begin{figure}
\centering
\includegraphics[width=14cm,keepaspectratio]{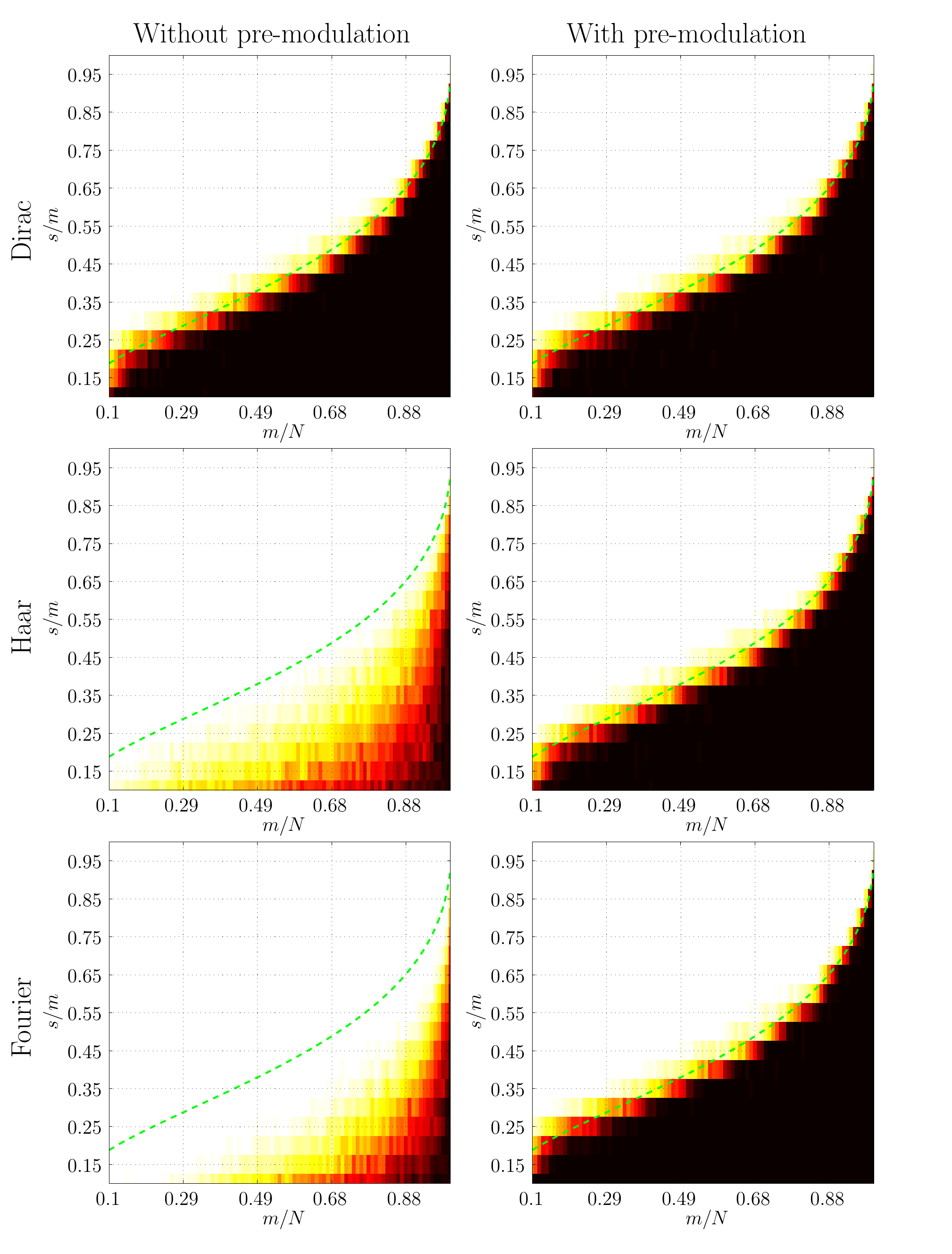}
\caption{\label{fig:phase transition hadamard}Phase transition of the $\ell_1$-minimization problem for different sparsity bases and random selection of \textbf{Hadamard} measurements without (left panels) and with (right panels) random modulation. The sparsity bases considered are the Dirac basis (top), the Haar wavelet basis (center), and the Fourier basis (bottom). The dashed green line indicates the phase transition of Donoho-Tanner \cite{donoho09}. The color bar goes from white to black indicating a probability of recovery from $0$ to $1$.}
\end{figure} 
%

%
%
%
%
\subsection{Settings}
For the first set of simulations, we consider the Dirac, Fourier and Haar wavelet bases as sparsity basis $\ma{\Psi}$ and choose the Fourier basis as the sensing matrix $\ma{\Phi}$. We generate complex $s$-sparse signals of size $N=1024$ with $s \in \left\{1, \ldots, N \right\}$. The positions of the non-zero coefficients are chosen uniformly at random in $\left\{1, \ldots, N \right\}$, their phases are set by generating a Steinhaus sequence, and their amplitudes follow a uniform distribution over $\left[0, 1 \right]$. The signals are then probed according to relation (\ref{eq:measurement model}) or (\ref{eq:spread measurement matrix}) and reconstructed from different number of measurements $m \in \left\{s, \ldots, 10s \right\}$ by solving the $\ell_1$-minimization problem (\ref{eq:BP}) with the SPGL$1$ toolbox \cite{berg08, spgl1}. For each pair $(m, s)$, we compute the probability of recovery\footnote{perfect recovery is considered if the $\ell_2$ norm between the original signal $\bm{x}$ and the reconstructed signal $\bm{x}^\star$ satisfy: $\norm{\bm{x}-\bm{x}^\star}_2\leq10^{-3}\norm{\bm{x}}_2$.} over $100$ simulations.

For the second set of simulations, the same protocol is applied with the same sparsity basis but with the Hadamard basis as the sensing matrix $\ma \Phi$.

%
%
%
%
\subsection{Results}

Figure \ref{fig:phase transition fourier} shows the phase transitions of the $\ell_1$-minimization problem obtained for sparse signals in the Dirac, Haar and Fourier sparsity bases and probed in the Fourier basis with and without random pre-modulation. Figure \ref{fig:phase transition hadamard} shows the same graphs but with measurements performed in the Hadamard basis.

In the absence of pre-modulation, one can note that the phase transitions depend on the mutual coherence of the system as predicted by Theorem \ref{th:standard uniform recovery}. For the pairs Fourier-Dirac and Hadamard-Dirac, the mutual coherence is optimal and the experimental phase transitions match the one of Donoho-Tanner (dashed green line) \cite{donoho09}. For all the other cases, the coherence is not optimal and the region where the signals are recovered is much smaller. The worst case is obtained for the pair Fourier-Fourier for which $\mu = 1$.

In the presence of pre-modulation, Corollary \ref{co:spread spectrum universality} predicts that the performance should not depend on the sparsity basis and should become optimal. It is confirmed by the phases transition showed on Figures \ref{fig:phase transition fourier} and \ref{fig:phase transition hadamard} as they all match the phase transition of Donoho-Tanner, even for the pair Fourier-Fourier!

%
%
%
%
\section{Application to realistic Fourier imaging}
\label{sec:analog spread spectrum}

In this section, we discuss the application of the spread spectrum technique to realistic analog Fourier imaging such as radio interferometric imaging or MRI. Firstly, we introduce the exact sensing matrix needed to account for the analog nature of the imaging problem. Secondly, while our original theoretical results strictly hold only in a digital setting, we derive explicit performance guarantees for the analog version of the spread spectrum technique. We also confirm on the basis of simulations that the spread spectrum technique drastically enhances the quality of reconstructed signals.

%
%
%
%
\subsection{Sensing model}

Radio interferometry dates back to more than sixty years ago \cite{ryle46, ryle59, ryle60, thompson01}. It allows observations of the sky with angular resolutions and sensitivities inaccessible with a single telescope. In a few words, radio telescope arrays synthesize the aperture of a unique telescope whose size would be the maximum projected distance between two telescopes of the array on the plane perpendicular to line of sight. Considering small field of views, the signal probed can be considered as a planar image on the plane perpendicular to the pointing direction of the instrument. Measurements are obtained through correlation of the incoming electric fields between each pair of telescopes. As showed by the van Cittert-Zernike theorem \cite{thompson01}, these measurements correspond to the Fourier transform of the image multiplied by an illumination function. In general, the number of spatial frequencies probed are much smaller than the number of frequencies required by the Nyquist-Shannon theorem, so that the Fourier coverage is incomplete. An ill-posed inverse problem is thus defined for reconstruction of the original image. To address this problem, approaches based on compressed sensing have recently been developed \cite{wiaux09, wenger10, li11}.

MR images are created by nuclear magnetic resonance in the tissues to be imaged. Standard MR measurements take the form of Fourier (also called $k$-space) coefficients of the image of interest. These measurements are obtained by application of linear gradient magnetic fields that provides the Fourier coefficient of the signal at a spatial frequency proportional the gradient strength and its duration of application. Accelerating the acquisition process, or equivalently increasing the achievable resolution for a fixed acquisition time, is of major interest for MRI applications. To address this problem, recent approaches based on compressed sensing seek to reconstruct the signal from incomplete information. In this context, several approaches have been designed \cite{lustig07, sebert08, liang09b, wang09, seeger10, haldar10, puy11b, wiaux10, puy11a}.

In light of the results of Section \ref{sec:discrete spread spectrum}, Fourier imaging is a perfect framework for the spread spectrum technique, apart from the analog nature of the corresponding imaging problems. In the quoted applications, the random pre-modulation is replaced by a linear chirp pre-modulation \cite{wiaux09b, wiaux09a, puy10a, wiaux10, puy11a, puy11b}. In radio interferometry, this modulation is inherently part of the acquisition process \cite{wiaux09b, wiaux09a}. In MRI, it is easily implemented through the use of dedicated coils or RF pulses \cite{puy11a, puy11b}. For two-dimensional signals, the linear chirp with chirp rate $w \in \Rbb$ reads as a complex-valued function $c : \tau \mapsto {\rm e}^{ {\rm i} \pi w \tau^2 }$ of the spatial variable $\tau \in\Rbb^2$. Note that for high chirp rates $w$, i.e. for chirp whose band-limit is of the same order of the band-limit of the signal, this chirp shares the following important properties with the random modulation: it is a wide-band signal which does not change the norm of the signal $x$, as $\abs{c(\tau)} = 1$ whatever $\tau \in \Rbb^2$.

In this setting, the complete linear relationship between the signal and the measurements is given by
\begin{eqnarray}
\label{eq:analog measurement matrix}
\bm y = \ma{A}^{w}_\Omega \, \bm \alpha \text{ with } \ma{A}^{w}_\Omega = \ma{F^*_\Omega CU \Psi} \in \Cbb^{m \times N}.
\end{eqnarray}
In the above equation, the matrix $\ma{U}$ represents an up-sampling operator needed to avoid any aliasing of the modulated signal due to a lack of sampling resolution in a digital description of the originally analog problem. The convolution in Fourier space induced by the analog modulation implies, in contrast with the digital setting studied before, that the band limit of the modulated signal is the sum of the individual band limits of the original signal and of the chirp $c$. { We assume here that, on its finite field of view $L$, the signal $x$ is approximately band-limited with a cut-off frequency at $B$, i.e., its energy beyond the frequency $B$ is negligible. The signal $x$ is thus discretized on a grid of $N = 2LB$ points.} On this field of view $L$, the linear chirp $c$ may be approximated by a band limited function of band limit identified by its maximum instantaneous frequency $\vert w \vert L/2$. This band limit can also be parametrized in terms of a discrete chirp rate $\bar{w} = w L^2/N$ and thus $\vert w \vert L/2 = \abs{\bar{w}} B$. Therefore, an up-sampled grid with at least $N_{w}=(1+\vert\bar{w}\vert)N$ points needs to be considered and the modulated signal is correctly obtained by applying the chirp modulation on the signal after up-sampling on the $N_{w}$ points grid\footnote{In a full generality, natural signals are not necessarily band-limited. The spread spectrum technique can easily be adapted to this case. The sensing model should simply be modified to account for the fact that, if measurements are performed at frequencies up to a band limit $B$, they unavoidably contain energy of the signal up to band limit $(1+\bar{w})B$.}. The up-sampling operator $\ma{U}$, implemented in Fourier space by zero padding, is of size ${N_{w}\times N}$ and satisfies $\ma{U}^*\ma{U} = \ma{I} \in \Cbb^{N \times N}$. Finally, the matrix $\ma{C}\in\Cbb^{N_{w}\times N_{w}}$ is the diagonal matrix implementing the chirp modulation on this up-sampled grid and the matrix $\ma{F} = \left(\bm{f}_i\right)_{1 \leq i \leq N_{w}}\in\Cbb^{N_{w}\times N_{w}}$ stands for the discrete Fourier basis on the same grid. The indices $\Omega = \left\{l_1, \ldots, l_m \right\}$ of the Fourier vectors selected to probe the signal are chosen independently and uniformly at random from $\left\{1, \ldots, N_{w}\right\}$.

%
%
%
%
\subsection{Illustration}

\begin{figure}
\centering
\includegraphics[width=16cm,keepaspectratio]{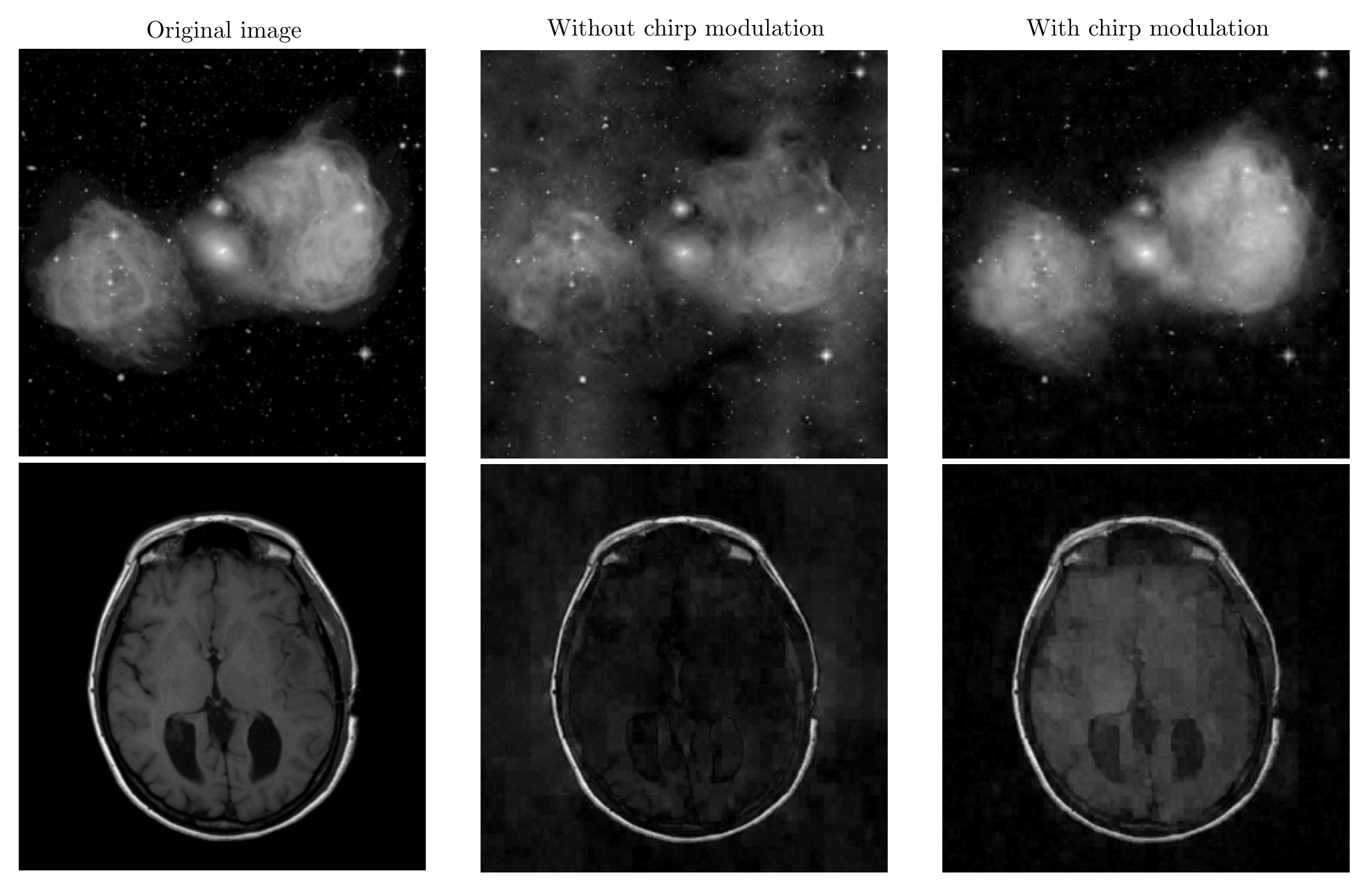}
\caption{\label{fig:galaxy} Top panels: Image of the giant elliptical galaxy NGC1316 (center of the image) \emph{devouring} its small northern neighbor. The image shows the radio emission associated with this encounter superimposed on an optical image. The radio emission was imaged using the Very Large Array in New Mexico (Image courtesy of NRAO/AUI and J. M. Uson). The image size is $N= 256 \times 256$. Bottom panels: MRI image of a brain from the B{\sc rainix} database. From left to right: original image; complex magnitudes of the reconstructed images from $m = 0.4 N$ measurements without chirp modulation; complex magnitudes of the reconstructed images from $m = 0.4 N$ measurements with chirp modulation.}
\end{figure} 

Up to the introduction of the matrix $\ma{U}$ and the substitution of the linear chirp modulation for the random modulation, we are in the same setting as the one studied in Section \ref{sec:discrete spread spectrum}. To illustrate the effectiveness of the spread spectrum technique, we consider two images of size $N = 256 \times 256$ showed in Figure \ref{fig:galaxy}. { The first image shows the radio emission associated with the encounter of a galaxy with its northern neighbor. It was acquired with the Very Large Array in New Mexico \cite{vla}. The second image shows a brain acquired in an MRI scanner. This image is part of the B{\sc rainix} database \cite{brainix}. These images are probed according to relation (\ref{eq:analog measurement matrix}) in the absence ($\bar{w} = 0$) and presence ($\bar{w} = 0.1$) of a linear chirp modulation. Independent and identically distributed Gaussian noise with zero-mean is also added to the measurements. The variance $\sigma^2$ of the noise is defined such that the input ${\rm snr} = - 10 \log_{10}\left(\sigma^2/\Sigma^2_{\bm{x}}\right)$ is $30\,{\rm dB}$ ($\Sigma_{\bm{x}}$ stands for the sampled standard deviation of $\bm{x}$). The images are reconstructed from $m = 0.4 N$ complex Fourier measurements by solving the $\ell_1$-minimization problem (\ref{eq:BPDN}). Note that in order to stay in the setting of the theorems presented so far, no reality constraint is enforced in the reconstructions, so that the reconstructed images are complex valued. The sparsity bases $\ma{\Psi}$ used are the Daubechies-6 and Haar wavelet bases for the galaxy and the brain respectively. In each case, $20$ reconstructions are performed for different noise and mask realizations. The complex magnitudes of reconstructed images with median mean squared errors are presented in Figure \ref{fig:galaxy}.}

In the absence of linear chirp modulation, the quality of the reconstructed image is very low. However, one can already note that the fine scale structures are much better reconstructed than those at large scales. The fine details live at the small scales of the wavelet decomposition whereas the large structures live at larger scales. The small scale wavelets being more incoherent with the Fourier basis than the larger wavelets, the high frequency details are naturally better recovered.

In the presence of the linear chirp modulation, all the wavelets in $\ma{\Psi}$ become optimally incoherent with the Fourier basis thanks to the universality of the spread spectrum technique. Consequently, as one can observe on Figure \ref{fig:galaxy}, the low and high frequency details are better reconstructed and the image quality is drastically enhanced. { Note that much better reconstructions can be obtained for the brain image by substituting the Total Variation norm\footnote{$\ell_1$ norm of the magnitude of the gradient.} for the $\ell_1$ norm in (\ref{eq:BPDN}) \cite{lustig07, puy11b}. However, Theorems \ref{th:standard uniform recovery} and \ref{th:spread spectrum uniform recovery} do not hold for such a norm.}

Let us acknowledge that these simulations are not fully realistic. For example, in radio-interferometry, the spatial frequency cannot be chosen at random. To simulate realistic acquisitions, one would have to consider non-random measurements in the continuous Fourier plane. Such a study is beyond the scope of this work. However, in the context of MRI, part of the authors implemented and tested this technique on a real scanner with \emph{in vivo} acquisitions \cite{puy11a, puy11b}.

%
%
%
%
\subsection{Modified recovery condition}

Because of the modifications introduced in (\ref{eq:analog measurement matrix}) to account for the analog nature of the problem, the digital theory associated with the measurement matrix (\ref{eq:spread measurement matrix}) does not explicitly apply. Nevertheless, the previous illustration shows that the spread spectrum technique is indeed still very effective in this analog setting. Actually, performance guarantees similar to Theorem \ref{th:standard uniform recovery} may be obtained in this setting. 

\begin{theorem}
\label{th:analognonunifrec}
Let $\ma{A}^{w} = \ma{F^* CU \Psi} \in \Cbb^{m \times N}$, $\mu_{w} = \max_{1 \leq i,j \leq N} \left| \scp{\bm{f}_i}{\ma{CU}\bm{\psi}_j}\right|$, $\bm{\alpha} \in \Cbb^N$ be an $s$-sparse vector, $\Omega = \left\{l_1, \ldots, l_m \right\}$ be a set of $m$ indices chosen independently and uniformly at random from $\left\{1, \ldots, N_{w}\right\}$, and $\bm{y} = \ma{A}^{\rm c}_\Omega \bm{\alpha} \in \Cbb^{m}$. For some universal constants $C>0$ and $\gamma>1$, if the number of measurements $m$ satisfies
\begin{eqnarray}
\label{eq:ananonunif}
m \geq C N_{w} \, \mu_{w}^2 s \log^4(N),
\end{eqnarray}
then $\bm \alpha$ is the unique minimizer of the $\ell_1$-minimization problem (\ref{eq:BP}) with probability at least $1-N^{-\gamma \log^3(N)}$.
\end{theorem}
\begin{proof}
The proof follows directly from Theorem $4.4$ in \cite{rauhut10}. Indeed, Theorem $4.4$ applies to any matrices $\ma{A}_\Omega$ associated to an orthonormal system (with respect to the probability measure used to draw $\Omega$) that satisfies the so-called boundedness condition (see Section $4.1$ of \cite{rauhut10} for more details). 

Let us denote $\tilde{\ma{A}}^{w} = (\tilde{a}_{ik})_{1 \leq i,k \leq N} = N_{w}^{1/2} \; \ma{A}^{w}$ the normalized measurement matrix. It is easy to check that this new matrix is orthonormal relative to the discrete uniform probability measure on $\left\{1, \ldots, N_{w}\right\}$. Indeed
\begin{eqnarray*}
\sum_{i=1}^{N_{w}} \tilde{a}_{ik}^* \tilde{a}_{ij} \, N_{w}^{-1} = 
\sum_{i=1}^{N_{w}} \left(N_{w}^{1/2} \bm{f}^*_i\ma{CU}\bm{\psi}_j\right)^* \, \left(N_{w}^{1/2} \bm{f}^*_i\ma{CU}\bm{\psi}_k\right) N_{w}^{-1} = \bm{\psi}_j^* \ma{U^* C^* FF^*CU}\bm{\psi}_k =\delta_{jk},
\end{eqnarray*}
as $\ma{U^*U} = \ma{I} \in \Cbb^{N \times N}$ and $\ma{C^*C} = \ma{FF^*} = \ma{I} \in \Cbb^{N_{w} \times N_{w}}$. Furthermore, the boundedness condition is satisfied as $\max_{1 \leq i, j \leq N} \vert \tilde{a}^{w}_{ij} \vert \leq N_{w}^{1/2} \mu_{w}$. Applying Theorem $4.4$ in \cite{rauhut10} to the matrix $\tilde{\ma{A}}^{w}$ terminates the proof.
\end{proof}

Note that Theorem $4.4$ in \cite{rauhut10} also ensures that our analog sensing scheme is stable relative to noise and non exact sparsity if condition (\ref{eq:ananonunif}) is satisfied. Also note that one can obtain a similar results using Theorem $1.1$ and $1.2$ of the very recently accepted paper \cite{candes11}.

In view of this theorem, one can notice that the number of measurements needed for accurate reconstructions of sparse signals is proportional to the sparsity $s$ times the product $N_{w}\,\mu_{w}^2$, with factors depending on the chirp rate $\bar{w}$. In the analog framework, two effects are actually competing. On the one hand, the mutual coherence $\mu_{w}^2$ of the system is decreasing with the spread spectrum phenomenon, but, on the other hand, the number of accessible frequencies $N_{w}$ that bear information is increasing linearly with the chirp rate. The optimal recovery conditions are reached for a chirp rate $\bar{w}$ that ensures that the product $N_{w}\,\mu_{w}^2$ is at its minimum. This minimum might change depending on the sparsity matrix, so the universality of the recovery is formally lost.

To illustrate this effect, Table \ref{tab:mu} shows values of the product $N_{w} \, \mu_{w}^2$ for different chirp rates $\bar{w}$ and two different sparsity matrices: the Fourier and the Dirac bases. The values are computed numerically for a size of signal $N=1024$. In the case of the Dirac basis, one can notice that the product $N_{w} \, \mu_{w}^2$ slightly increases with the chirp rate, thus predicting that the performance should even slightly diminish in the presence of a chirp. On the contrary, the product is drastically reduced for the Fourier basis as $\bar{w}$ increases, predicting a significant performance improvement in the presence of chirp modulation.

%
%
%
%
\subsection{Experiments}

\begin{figure}
\centering
\includegraphics[width=15cm,keepaspectratio]{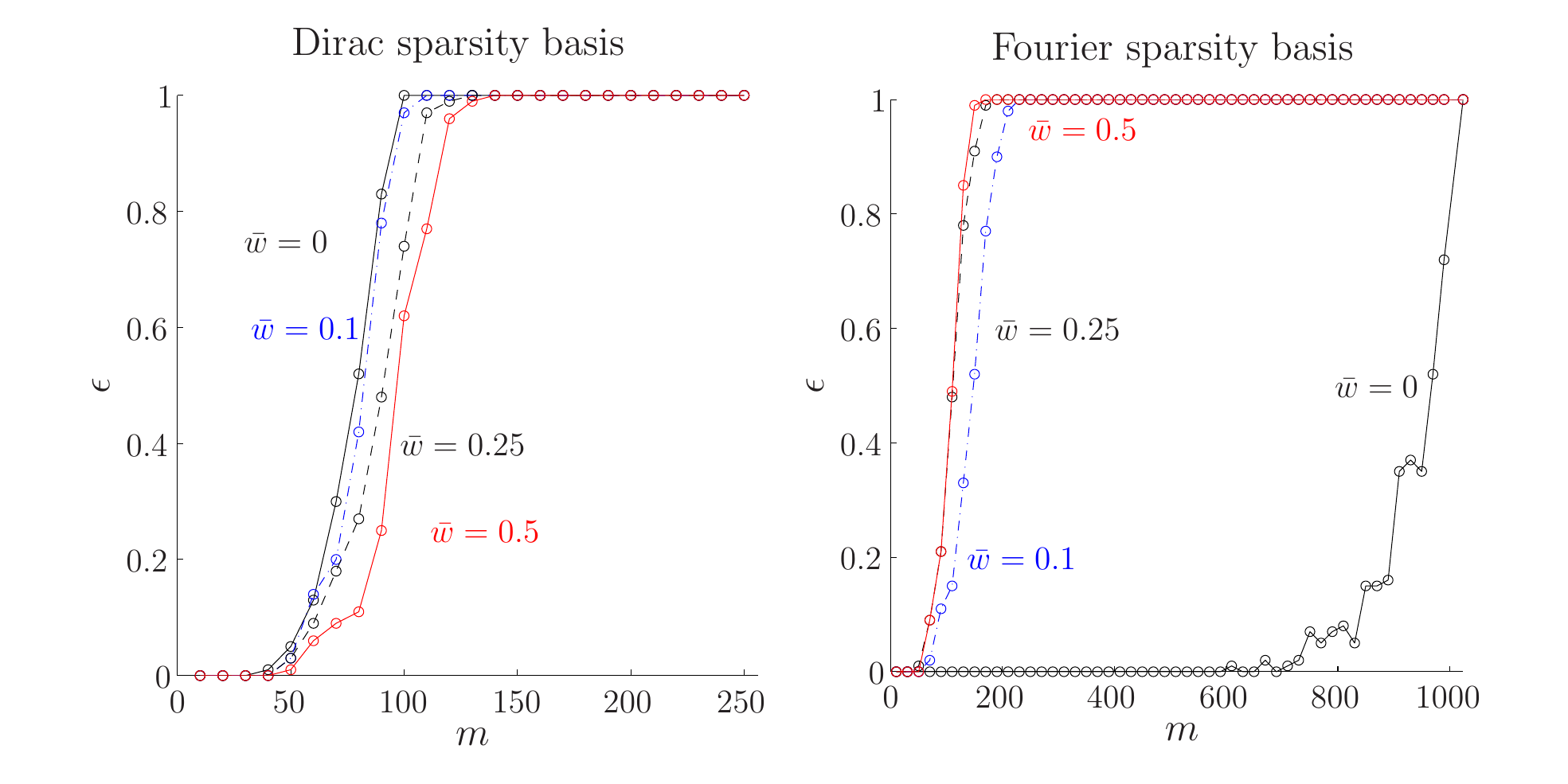}
\caption{\label{fig:prob recovery analog}Probability of recovery $\epsilon$ of $10$-sparse signals as a function of the number measurement $m$ obtained with the measurement matrix (\ref{eq:analog measurement matrix}) for two different sparsity basis: the Dirac basis (left) and the Fourier basis (right). The continuous black curve corresponds to the probability of recovery for $\bar{w} = 0$. The dot-dashed blue curve corresponds to the probability of recovery for $\bar{w} = 0.1$. The dashed black curve corresponds to the probability of recovery for $\bar{w} = 0.25$. The continuous red curve corresponds to the probability of recovery for $\bar{w} = 0.5$.}
\end{figure} 
\begin{table}[b]
\renewcommand{\arraystretch}{1.3}
\centering
\begin{tabular}{ccc}
\hline 
Sparsity basis 	& Dirac & Fourier \\
\hline\hline
$N_{w} \, \mu_{w}^2$ at $\bar{w}=0$ & $1.00$ &  $1.02 \cdot 10^3$\\
$N_{w} \, \mu_{w}^2$ at $\bar{w}=0.10$ & $2.58$ &  $1.54 \cdot 10^1 $  \\
$N_{w} \, \mu_{w}^2$ at $\bar{w}=0.25$ & $3.15$ & $6.95$ \\
$N_{w} \, \mu_{w}^2$ at $\bar{w}=0.50$ & $3.46$ & $4.13$ \\
\hline
\end{tabular}
\caption{\label{tab:mu} Influence of a chirp modulation on $N_{w} \, \mu_{w}^2$.}
\end{table}

To confirm the theoretical predictions of the previous section, we consider the Dirac and Fourier bases as sparsity matrices $\ma \Psi$. We then generate complex $s$-sparse signals of size $N=1024$ with $s = 10$. The positions of the non-zero coefficients are chosen uniformly at random in $\left\{1, \ldots, N \right\}$, their signs are set by generating a Steinhaus sequence, and their amplitudes follow a uniform distribution in $\left[0, 1\right]$. The signals are then probed according to relation (\ref{eq:analog measurement matrix}) and reconstructed from different number of measurements $m \in \left\{s, \ldots, N \right\}$ by solving the $\ell_1$-minimization problem (\ref{eq:BP}) with the SPGL$1$ toolbox. For each pair $(m, s)$, we compute the probability of recovery over $100$ simulations for different chirp rate $\bar{w} \in \left\{0, 0.1, 0.25, 0.5 \right\}$. Figure \ref{fig:prob recovery analog} shows the probability of recovery $\epsilon$ as a function of the number of measurements.

Firstly, in the case where $\ma \Psi$ is the Dirac basis, one can notice that the number of measurements needed to reach a probability of recovery of $1$ slightly increases with the chirp rate $\bar{w}$. This is in line with the value in Table \ref{tab:mu} and Theorem \ref{th:analognonunifrec}. 

Secondly, in the case where $\ma \Psi$ is the Fourier basis, the performance becomes much better in the presence of a chirp. As predicted by the value in Table \ref{tab:mu}, the improvement is drastic when $\bar{w}$ goes from $0$ to $0.1$ and then starts to saturate between $0.1$ and $0.5$.

Thirdly, according to Table \ref{tab:mu}, the product $N_{w} \, \mu_{w}^2$ is equal to $4.13$ for the Fourier basis when $\bar{w}=0.5$. This is nearly the value obtained with the Dirac basis for the same chirp rate, suggesting the same probability of recovery for the same number of measurements. Indeed, one can notice on Figure \ref{fig:prob recovery analog} that the number of measurement needed to reach a probability of recovery of $1$ is around $100$ in both cases. 

Finally, these results also suggest that the spread spectrum technique in the modified setting is almost universal in practice. Indeed, for the perfectly incoherent pair Fourier-Dirac of sensing-sparsity bases, the number of measurements needed for perfect recovery is around $100$ and this number remains almost unchanged in presence of the linear chirp modulation. Furthermore, for the pair Fourier-Fourier, the spread spectrum technique allows to reduce the number of measurements for perfect recovery close to this optimal value.

%
%
%
%
\section{Conclusion}
\label{sec:conclusion}

We have presented a compressed sensing strategy that consists of a wide bandwidth pre-modulation of the signal of interest before projection onto randomly selected vectors of an orthonormal basis. In a digital setting with a random pre-modulation, the technique was proved to be universal for sensing bases such as the Fourier or Hadamard bases, where it may be implemented efficiently.  Our results were confirmed through a numerical analysis of the phase transition of the $\ell_1$-minimization problem for different pairs of sensing and sparsity bases. 

The spread spectrum technique was also shown to be of great interest for realistic analog Fourier imaging. In applications such as radio interferometry and MRI, the originally digital random pre-modulation may be mimicked by an analog linear chirp. Explicit performance guarantees for the analog version of the technique with a chirp modulation were derived. It shows that recovery results are still enhanced in this setting, though universality does not strictly hold anymore. Numerical simulations have shown that the quality of reconstructed signals is drastically enhanced in this more realistic setting, also for pairs of sensing-sparsity bases initially highly coherent, such as the Fourier-Fourier pair.

%
%
%
%
\section{Acknowledgments}
\ifthenelse{\boolean{publ}}{\small}{}
This work was supported in part by the Center for Biomedical Imaging (CIBM) of the Geneva and Lausanne Universities, EPFL, and the Leenaards and Louis-Jeantet foundations, in part by the Swiss National Science Foundation (SNSF) under grant PP00P2-123438, also by the EU FET-Open project FP7-ICT-225913-SMALL: Sparse Models, Algorithms and Learning for Large-Scale data, and by the EPFL-Merck Serono Alliance award.

%
%
%
%
\bibliographystyle{bmc_article}  
\bibliography{EURASIP-puyetal-Spread_spectrum_universality}
\newpage
\ifthenelse{\boolean{publ}}{\end{multicols}}{}

\end{bmcformat}
\end{document}